\documentclass[12pt]{amsart}

\usepackage{amsmath,amssymb}

\begin{document}

\newtheorem{theorem}{Theorem}
\newtheorem{proposition}[theorem]{Proposition}
\newtheorem{corollary}[theorem]{Corollary}
\newtheorem{lemma}[theorem]{Lemma}
\newtheorem{definition}[theorem]{Definition}
\newtheorem{remark}[theorem]{Remark}
\newtheorem{conjecture}[theorem]{Conjecture}

\newcommand\Diff{\text{Diff}([0,1])}
\newcommand\R{\mathbb R}
\newcommand\D{\mathcal D}
\newcommand\id{\text{id}}

\

\title
[Time warping with Hellinger elasticity]
{Time warping with Hellinger elasticity}
\author{Yuly Billig}
\address{School of Mathematics and Statistics, Carleton University, Ottawa, Canada}
\email{billig@math.carleton.ca}

\maketitle


\begin{abstract} We consider a matching problem for time series with values in an arbitrary metric space, with the stretching penalty given by the Hellinger kernel. To optimize this matching, we introduce the Elastic Time Warping algorithm with a cubic computational complexity. 
\end{abstract}

\section{Introduction}

Functional data analysis has a multitude of applications: speech recognition \cite{SC}, biomedical \cite{K}, motion analysis \cite{SJW}, economics \cite{MMQV}, DNA matching \cite{H}, to name a few.

In this paper, we discuss a new time warping algorithm for matching time series with a Hellinger penalty for stretching time. Fr\'echet metric gives a way to compute a distance between two curves $f, g$ with values in a metric space $(X, \rho)$, which is independent of a parametrization: 
\[\inf_{\alpha,\beta\in\D} \sup_{\tau \in [0,1]} \rho(f(\alpha(\tau)), g(\beta(\tau))).\]
Here $\D$ is the group of diffeomorphisms of $[0,1]$. In many situations, however, we care about a specific time parametrization, and we want to impose a penalty for changing the parametrization of the curve.

This can be done with the Skorohod metric \cite{S, B}:
\[\inf_{\alpha\in\D} \max \left(  \sup_{\tau \in [0,1]} |\alpha(\tau), \tau|, \sup_{\tau \in [0,1]} \rho(f(\alpha(\tau)), g(\tau)) \right).\]

An important method in functional data analysis is based on the square root velocity framework 
\cite{S, Br, SK}, where the principal object of analysis is the square root velocity function 
$f^\prime(t) / \sqrt{\left| f^\prime(t) \right|}$. In this framework, the matching is performed using
\[\sup_{\alpha, \beta \in \D} \int_0^1 
\frac{f^\prime(\alpha(\tau))}{\sqrt{\left| f^\prime(\alpha(\tau)) \right|}} 
\frac{g^\prime(\beta(\tau))}{\sqrt{\left| g^\prime(\beta(\tau)) \right|}} 
\sqrt{\alpha^\prime(\tau)} \sqrt{\beta^\prime(\tau)} \,d\tau.\]
A connection of this framework with the Hellinger distance is pointed out in \cite{SK}.
Observe that for a reparametrization $\alpha \in \D$, the derivative $\alpha^\prime$ has all the properties of the density function of a probability distribution: it is positive 
and its integral is equal to 1. This allows us to borrow the tools from probability theory and apply them to the group $\D$ of diffeomorphisms of the segment $[0,1]$. For two diffeomorphisms $\alpha, \beta \in \D$, we define the Hellinger similarity coefficient
\[C(\alpha, \beta) = \int_0^1 \sqrt{\alpha^\prime (t)} \sqrt{\beta^\prime (t)} \,dt,\]
which gives rise to the Hellinger distance
\[\theta(\alpha, \beta) = arccos \, C(\alpha, \beta).\]

This leads to the following metric on the space of functions (Proposition \ref{Hmetric}):
\[d(f, g) = \inf_{\alpha,\beta\in\D}  \left( \theta(\alpha, \beta) + \sup_{\tau \in [0,1]} \rho(f(\alpha(\tau)), g(\beta(\tau))) \right).\]
In this metric, the penalty for stretching is given by the Hellinger term.

Commonly used objective functions for dynamic time warping are reviewed in \cite{HMB}.

For many applications, such as DNA matching, we are interested in knowing how close the matching pieces are, rather than how different are the non-matching pieces. For this reason, we may be more interested in the similarity coefficient $K(f,g)$, which has values between $0$ and $1$, and is equal to $1$ whenever $f = g$. For clustering algorithms, it is important to know the nearest neighbours of a given point, and a similarity coefficient can provide this information, effectively replacing the metric.

Motivated by the square root velocity framework and formula for the Hellinger metric, we introduce the following similarity coefficient:
\[K (f, g) = \sup_{\alpha, \beta \in \D} \int_0^1 \exp\big(-\rho(f(\alpha(\tau)), g(\beta(\tau))\big) \sqrt{\alpha^\prime(\tau)} \sqrt{\beta^\prime(\tau)} \,d\tau.\]
Unlike the square root velocity framework, which is applicable to functions with values in a vector space, the above similarity coefficient may be used for functions with values in arbitrary metric spaces.

A time series may be viewed as a piecewise constant function. Then the above expression yields a similarity coefficient for the time series.
For the given time series $f$, $g$, we determine the optimal parametrizations, which maximize the above integral. 
This leads to an algorithm for computing this similarity coefficient. The Elastic Time Warping algorith which we describe, gives 
the best matching between two time series, taking into account the penalty for stretching. If the time series $f$, $g$ have lengths $n$ and $m$ respectively, the computational complexity of the Elastic Time Warping algorithm is $(n+m)nm$ and the memory requirement is $n m$.

The paper is organized as follows. In Section \ref{Hellinger}, we discuss the properties of the Hellinger metric on the group of diffeomorphisms. In Section \ref{distance}, we use the metric on diffeomorphisms to define a metric on the space of functions with the Hellinger penalty for stretching. Finally, in Section \ref{algorithm}, we introduce the Elastic Time Warping algorithm that optimizes matching between two time series, based on a similarity coefficient with the Hellinger kernel.

{\bf Acknowledgements.} The author gratefully acknowledges support with a Discovery grant from the Natural Sciences and Engineering Research Council of Canada.

\section{Hellinger metric on $\Diff$.}
\label{Hellinger}

Consider the group of orientation-preserving diffeomorphisms of the segment $[0,1]$:
\[\Diff = \left\{ \alpha \in C^1([0,1]) \, | \, \alpha(0) = 0, \alpha(1) = 1, \alpha^\prime(t) > 0 \ \forall t \in [0,1] \right\}.\]
Note that a diffeomorphism $\alpha \in \Diff$ has the properties of a cumulative distribution function, and its derivative
$\alpha^\prime$ may be viewed as a probability density function. Because of this, we can borrow elements of probability theory when we study this group.

For $\alpha, \beta \in \Diff$, the Hellinger similarity coefficient is defined as 
\[C(\alpha, \beta) = \int_0^1 \sqrt{\alpha^\prime (t)} \sqrt{\beta^\prime (t)} \,dt.\]
Since $\int_0^1 \alpha^\prime (t) dt = 1$, function $\sqrt{\alpha^\prime (t)}$ is a unit vector in Hilbert space $L^2([0,1])$.

There are three common variants of the Hellinger distance. First, one defines 
\[\theta(\alpha, \beta) = arccos \, C(\alpha, \beta).\]
This is the angle between $\sqrt{\alpha^\prime}$ and  $\sqrt{\beta^\prime}$ on the unit sphere in Hilbert space, which is a distance function, being the geodesic distance on the sphere. It follows from Cauchy-Schwarz inequality that $0 < C(\alpha, \beta) \leq 1$ and $C(\alpha, \beta) = 1$ if and only if $\alpha = \beta$. Thus $0 \leq \theta(\alpha, \beta) < \frac{\pi}{2}$.

The following result on properties of the distance function is elementary \cite{Co}:

\begin{lemma}
\label{distance}
Let $X$ be a metric space with distance function $\rho: X \times X \rightarrow S \subset \R_{\geq 0}$. Suppose function $F: \R_{\geq 0} \rightarrow \R_{\geq 0}$
satisfies $F(a+b) \leq F(a) + F(b)$ for all $a, b \in S$. Then $F \circ \rho$ is also a metric on $X$.
\end{lemma}

We note that sine function satisfies $\sin(a+b) \leq \sin(a) + \sin(b)$ for all $a, b \in [0, \pi/2]$. This follows immediately from the trigonometric identity for $\sin(a+b)$. By the above Lemma, 
\[S(\alpha, \beta) = \sin \theta(\alpha, \beta) = \sqrt{ 1 - C(\alpha, \beta)^2}\]
is also the distance function, i.e., satisfies the triangle inequality.

For the same reasons as above $\frac{1}{2} \theta(\alpha, \beta)$ and $\sin \frac{\theta(\alpha, \beta)}{2}$ are distance functions as well.
Applying trigonometric identity $(1 - \cos \theta)/2 = \sin^2 \frac{\theta}{2}$, we conclude that
\[H(\alpha, \beta) = \sqrt{ 1 - C(\alpha, \beta)}\]
is also a distance function on $\Diff$. 

An important property of the Hellinger similarity coefficient is its invariance under reparametrizations.

\begin{proposition}
\label{invariance}
For $\alpha, \beta, \gamma \in \Diff$
\[C(\alpha \circ \gamma, \beta \circ \gamma) = C(\alpha, \beta).\]
\end{proposition}
\begin{proof}
\begin{align*}
&\int_0^1 \sqrt{\frac{d}{dt} \alpha(\gamma(t))} \sqrt{\frac{d}{dt} \beta(\gamma(t))} \, dt \\
= &\int_0^1 \sqrt{\alpha^\prime(\gamma(t))} \sqrt{\beta^\prime(\gamma(t))} \gamma^\prime(t) \, dt \\
= &\int_0^1 \sqrt{\alpha^\prime(s)} \sqrt{\beta^\prime(s)} \, ds .
\end{align*}
\end{proof}

\begin{corollary}
Metrics $\theta(\alpha, \beta), S(\alpha, \beta)$, and $H(\alpha, \beta)$ on $\Diff$ are right-invariant. 
\end{corollary}

\begin{corollary}
\label{hop}
Let $\alpha, \beta \in \Diff$. Then

(a) $C(\alpha, \beta) = C(\alpha \circ \beta^{-1}, \id)$,

(b) $C(\alpha, \id) = C(\alpha^{-1}, \id)$.
\end{corollary}

\section{Distance on the space of functions with a Hellinger penalty for stretching}
\label{distance}

In this section, we consider distance on the space of functions with the Hellinger penalty for stretching.

Consider the space of bounded functions $B([0,1], X)$ defined on the segment $[0,1]$ with values in a metric space $(X, \rho)$. Let $\D = \Diff$. For $f, g \in B([0,1], X)$, the Fr\'echet distance is defined as 
\[\inf_{\alpha,\beta\in\D} \sup_{\tau \in [0,1]} \rho(f(\alpha(\tau)), g(\beta(\tau))).\]

In Fr\'echet distance, time may be freely stretched.
Let us define a new distance function on $B([0,1], X)$ by incorporating a penalty for stretching with a  Hellinger distance term. We denote by $D(\alpha, \beta)$ one of the distances $\theta, S$, or $H$, discussed in the previous section. Set
\begin{equation}
\label{metric}
d(f, g) = \inf_{\alpha,\beta\in\D}  \left( D(\alpha, \beta) + \sup_{\tau \in [0,1]} \rho(f(\alpha(\tau)), g(\beta(\tau))) \right).
\end{equation}

Proposition \ref{invariance} implies invariance of $d$ under reparametrizations:
\begin{lemma}
For $\gamma \in \D$ we have $d(f \circ \gamma, g \circ \gamma) = d(f, g)$.
\end{lemma}

Applying Corollary \ref{hop}(a), we see that we can define $d(f,g)$ in a slightly simpler, but less symmetric way:
\[d(f, g) = \inf_{\alpha\in\D}  \left( D(\alpha, \id) + \sup_{\tau \in [0,1]} \rho(f(\alpha(\tau)), g(\tau)) \right).\]

Compare this with the definition of Skorohod metric(s) \cite{B}:
\[\inf_{\alpha\in\D} \max \left(  \sup_{\tau \in [0,1]} |\alpha(\tau), \tau|, \sup_{\tau \in [0,1]} \rho(f(\alpha(\tau)), g(\tau)) \right),\]
or its variant
\[\inf_{\alpha\in\D} \left(  \sup_{\tau \in [0,1]} |\alpha(\tau), \tau| + \sup_{\tau \in [0,1]} \rho(f(\alpha(\tau)), g(\tau)) \right).\]

\begin{proposition}
\label{Hmetric}
Formula (\ref{metric}) defines a metric on $B([0,1], X)$.
\end{proposition}
\begin{proof}
It is immediate to see that $d(f, g)$ is symmetric: $d(f,g) = d(g,f)$. Also, $d(f,g) = 0$ if and only if $f = g$.
Let us prove the triangle inequality:
\[d(f,h) \leq d(f,g) + d(g,h).\]
By definition of the infimum, for every $\epsilon > 0$ there exist $\alpha_0, \beta_0 \in \D$, such that 
\begin{align*}
d(f,g) + d(g,h) > &\, D(\alpha_0, \id) + D(\id, \beta_0) \\
+ &\sup_{\tau \in [0,1]} \rho(f(\alpha_0(\tau)), g(\tau))
+ \sup_{\tau \in [0,1]} \rho(g(\tau), h(\beta_0(\tau))) - \epsilon \\
\geq & \, D(\alpha_0, \beta_0) + \sup_{\tau \in [0,1]} \rho(f(\alpha_0(\tau)), h(\beta_0(\tau))) - \epsilon \\
\geq & \, d(f,h) - \epsilon.
\end{align*}
Since $\epsilon$ is arbitrary, we get the desired inequality.
\end{proof}

In certain applications, such as DNA matching, we are interested in recording how close the matching pieces are, rather than how far apart are the non-matching pieces. For this reason, we may more be interested in computing the similarity coefficient, rather than the distance. 

A similarity coefficient is a function $C(f, g)$ with values $0 \leq C(f,g) \leq 1$ and $C(f, g) = 1$ whenever $f = g$. Intuitively, a similarity coefficient is a decreasing function of the distance.
 For clustering of data, we are often interested in the nearest neighbours of a given point, rather than the actual values of the distance. A similarity coefficient may be used for this purpose. 

Let $C([0,1], X)$ be a class of bounded functions on $[0,1]$ with values in a metric space $X$, such that $\rho(f(\tau), g(\tau))$ has a finite number of discontinuities. 
For $f, g \in C([0,1], X)$, we define the similarity coefficient
\[K (f, g) = \sup_{\alpha, \beta \in \D} \int_0^1 \exp\big(-\rho(f(\alpha(\tau)), g(\beta(\tau))\big) \sqrt{\alpha^\prime(\tau)} \sqrt{\beta^\prime(\tau)} \,d\tau.\]

Note that function $F(x) = 1 - e^{-x}$ satisfies the condition of Lemma \ref{distance}, so $1 - \exp(-|x-y|)$ is a metric on $\R$, with values close to $1$ for points that are far apart.

Another example of a metric with values between $0$ and $1$ is the Jaccard distance between finite sets. For such metrics, the triangle inequality has a limited value. Suppose text $B$ has 40\% common sentences with text $A$ and also with text $C$. This indicates a close connection between these pairs of texts. However, this does not imply any connection between texts $A$ and $C$.

\begin{lemma}
(a) $K(f, g) = K(g, f)$,

(b) $0 < K(f, g) \leq 1$,

(c) $K(f, g) = 1$ if and only if $f = g$ almost everywhere.
\end{lemma}
\begin{proof}
Part (a) is obvious.  Since $\rho(f(\tau), g(\tau))$ is bounded, say by constant $L > 0$, then taking $\alpha = \beta = \id$, we see that $K(f,g) \geq e^{-L} > 0$. On the other hand, $ \exp\big(-\rho(f(\alpha(\tau)), g(\beta(\tau))\big) \leq 1$, thus 
\[K(f,g) \leq \sup_{\alpha, \beta \in \D} \int_0^1 \sqrt{\alpha^\prime(\tau)} \sqrt{\beta^\prime(\tau)} \,d\tau \leq 1.\]
This proves part (b). For the equality $K(f,g) = 1$ to hold, we must have $\alpha = \beta$ and 
$\rho(f(\tau), g(\tau)) = 0$, except for a finite number of points. 
\end{proof}

Since the integral in the formula for $K(f,g)$ is invariant under reparametrizations: $(\alpha, \beta) \mapsto 
(\alpha \circ \gamma, \beta \circ \gamma)$, we can write the formula for $K(f,g)$ in a simpler, but less symmetric way:
\[K (f, g) = \sup_{\alpha \in \D} \int_0^1 \exp\big(-\rho(f(\alpha(\tau)), g(\tau)\big) \sqrt{\alpha^\prime(\tau)} \,d\tau.\]

More generally, given a similarity coefficient $C$ on set $X$, we can define the following similarity coefficient on the space of functions on $[0,1]$ with values in $X$:
\[K(f,g) = \sup_{\alpha, \beta \in \D} \int_0^1 C\big(f(\alpha(\tau)), g(\beta(\tau))\big) \sqrt{\alpha^\prime(\tau)} \sqrt{\beta^\prime(\tau)} \,d\tau.\]

\section{Elastic Time Warping Algorithm}
\label{algorithm}

A time-warping algorithm based on edit distance was introduced in \cite{M}. Various versions of the dynamic time warping algorithm are reviewed in \cite{DB}.

In this section, we introduce an algorithm for computing the similarity coefficient $K(f,g)$, discussed in the previous section.  We assume that we are given two time series, $\{ (f_i, s_i) \,|\, 0\leq i < n \}$ and 
$\{ (g_j, t_j) \,|\, 0\leq j < m\}$ with \[0 = s_0 < s_1 < \ldots < s_{n-1} < 1, \quad  0 = t_0 < t_1 < \ldots < t_{m-1} < 1,\]
and $f_i, g_j$ having values in set $X$ with a given similarity coefficient $C(x,y)$, $x,y \in X$.
We interpret these time series as piecewise constant functions on $[0,1]$, where $g(t) = g_j$ for $t_j \leq t < t_{j+1}$, and similarly for $f$. 

Fix $\alpha, \beta \in \D$. Let $\tau_i, \mu_j \in [0,1]$ be such that $\alpha(\tau_i) = s_i$, $\beta(\mu_j) = t_j$. Then the pair $(\alpha, \beta)$ defines an interlacing pattern on time series $\{ f_i \}$ and $\{ g_j \}$. Note that this interlacing pattern is invariant under reparametrizations  $(\alpha, \beta) \mapsto 
(\alpha \circ \gamma, \beta \circ \gamma)$. Moreover, any placement of points $\{ \tau_i \}, \{ \mu_j \}$ with the same interlacing pattern may be obtained using a reparametrization of $(\alpha, \beta)$. For this reason, we will be paying attention to the interlacing pattern, and then select the optimal $(\alpha, \beta)$ with the specific interlacing pattern.

Since it is sufficient to restrict to the case $\beta = \id$, we are going to do just that, so the points $\{ t_j \}$ will be fixed, and we will interlace points $\{ \tau_i \}$ between them.

In the Dynamic Time Warping algorithm, we are inductively matching $f_i$ with $g_j$, advancing by $1$ at each step, either index $i$, index $j$, or both. In our algorithm, we will be doing essentially the same, but we will also specify time locations for these values. 

We will impose the following natural restriction on the placement of $\{ \tau_i \}$. If for some pair of indices $i, j$ we have $\tau_i = t_j$ then either $\tau_{i+k} = t_{j+1}$ for some $k \geq 1$, and 
$t_j < \tau_{i+1} < \ldots < \tau_{i+k-1} < t_{j+1}$,  or $\tau_{i+1} = t_{j+k}$ for some $k \geq 1$.
The first scenario corresponds to matching $f_i, \ldots, f_{i+k-1}$ to $g_j$, the second case is matching 
$g_j, \ldots, g_{j+k-1}$ to $f_i$, and the case of $\tau_{i+1} = t_{j+1}$ corresponds to advancing both indices in DTW. Note that $\tau_0 = t_0 = 0$, to start the inductive process.

Let us prove three theoretical results about the optimal profile of the function $\alpha$, and about the optimal placement of $\{ \tau_i \}$.


\begin{proposition}
\label{linear}
Let $[a,b] = [\tau_i, \tau_{i+1}] \cap [t_j, t_{j+1}]$. 
Then the optimal parametrization $\alpha$ is a linear function on $(a,b)$.
\end{proposition}
\begin{proof}
Fix values $\alpha(a), \alpha(b)$ and freely vary $\alpha(\tau)$, $\tau\in [a,b]$ subject to these boundary constraints and keeping $\alpha^\prime(\tau) > 0$.  Note that 
$C(f(\alpha(\tau)), g(\tau))$ has constant value $C(f_i, g_j)$ on $(a,b)$, so we are just maximizing
\[\int_{a}^{b} \sqrt{\alpha^\prime(\tau)} \,d\tau.\]  Applying the Cauchy-Schwarz inequality, we get that 
\begin{equation*}
\left( \int_{a}^{b} \sqrt{\alpha^\prime(\tau)} \,d\tau \right)^2 \leq
\int_{a}^{b} \alpha^\prime(\tau) \,d\tau \int_{a}^{b} 1 \,d\tau 
= (b-a) (\alpha(b) - \alpha(a)),
\end{equation*}
and the equality holds when $\alpha^\prime (\tau)$ is a constant function on $(a, b)$.
\end{proof}

\begin{proposition}
\label{caseB}
Assume $\tau_i = t_j$ and $\tau_{i+1} = t_{j+k}$ for $k \geq 1$. Then for the optimal parametrization $\alpha$ with these constraints,
\begin{align*}
\int_{t_j}^{t_{j+k}}  C\big(f(\alpha(\tau)), g(\tau)\big) & \sqrt{\alpha^\prime(\tau)}  \,d\tau \\
= &\sqrt{(s_{i+1} - s_i) \sum_{\ell = 0}^{k-1} (t_{j+\ell+1} - t_{j+\ell}) C(f_i, g_{j+\ell})^2 }.
\end{align*}
\end{proposition}
\begin{proof}
By Proposition \ref{linear}, optimal parametrization $\alpha$ is piecewise linear. Denote by $k_\ell$ the slope of $\alpha(\tau)$ on interval $(t_{j+\ell}, t_{j+\ell+1})$. We want to maximize
\begin{equation*}
\sum_{\ell = 0}^{k-1} (t_{j+\ell+1} - t_{j+\ell}) C(f_i, g_{j+\ell}) \sqrt{k_\ell}
\end{equation*}
subject to the constraint
\begin{equation*}
\sum_{\ell = 0}^{k-1} (t_{j+\ell+1} - t_{j+\ell}) k_\ell = \int_{\tau_i}^{\tau_{i+1}} \alpha^\prime (\tau) \,d \tau = s_{i+1} - s_i.
\end{equation*}
For this, consider a positive definite scalar product on $\R^k$:
\begin{equation*}
({\bf x}, {\bf y}) = \sum_{\ell = 0}^{k-1} (t_{j+\ell+1} - t_{j+\ell}) \,x_\ell \,y_\ell,
\end{equation*}
and apply the Cauchy-Schwarz inequality for this product:
\begin{align*}
&\left( \sum_{\ell = 0}^{k-1}  (t_{j+\ell+1} - t_{j+\ell}) C(f_i, g_{j+\ell}) \sqrt{k_\ell} \right)^2 \\
&\quad\quad \leq \left( \sum_{\ell = 0}^{k-1} (t_{j+\ell+1} - t_{j+\ell}) k_\ell  \right) 
\left( \sum_{\ell = 0}^{k-1} (t_{j+\ell+1} - t_{j+\ell})  C(f_i, g_{j+\ell})^2 \right),
\end{align*}
and the equality holds when the slopes $k_\ell$ are proportional to $C(f_i, g_{j+\ell})^2$. The claim of the Proposition follows.
\end{proof}

\begin{proposition}
\label{caseA}
Assume $\tau_i = t_j$ and  $\tau_{i+k} = t_{j+1}$ for $k \geq 1$, and 
$t_j < \tau_{i+1} < \ldots < \tau_{i+k-1} < t_{j+1}$ Then for the optimal parametrization $\alpha$ with these constraints,
\begin{align*}
\int_{t_j}^{t_{j+1}}  C\big(f(\alpha(\tau)), g(\tau)\big) & \sqrt{\alpha^\prime(\tau)}  \,d\tau \\
= &\sqrt{(t_{j+1} - t_j) \sum_{\ell = 0}^{k-1} (s_{i+\ell+1} - s_{i+\ell}) C(f_{i + \ell}, g_j)^2 }.
\end{align*}
The optimal placement of points $\tau_{i+1}$, \ldots, $\tau_{i+k-1}$ is such that the lengths of subintervals
$\tau_{i+\ell+1} - \tau_{i+\ell}$ are proportional to $C(f_{i + \ell}, g_j)^2$.
\end{proposition}

The proof of this Proposition is analogous to the proof of Proposition \ref{caseB}, and we omit it.

Let us now describe the Elastic Time Warping algorithm. We are given two time series $\{ (f_i, s_i) \,|\, 0 \leq i < n \}$ and 
$\{ (g_j, t_j) \,|\, 0 \leq j < m \}$ with $f_i, g_j \in X$ and $0 = s_0 < s_1 < \ldots < s_{n-1} < 1$, \ $0 = t_0 < t_1 < \ldots < t_{m-1} < 1$. We are also given the similarity coefficient $C(f_i, g_j)$. Set $s_n = t_m = 1$.
We turn the time series into piecewise constant functions $f(s), g(t)$ on $[0,1]$ and we would like to maximize
\begin{equation}
\label{integral}
\int_0^1 C(f(\alpha(\tau)), g(\tau)) \sqrt{\alpha^\prime(\tau)} \,d \tau.
\end{equation}  
We impose the following restriction on the class of parametrizations $\alpha \in \D$ that we consider:
if $\alpha(t_j) = s_i$ then either $\alpha(t_{j+1}) = s_{i+k}$, or $\alpha(t_{j+k}) = s_{i+1}$ for some $k \geq 1$. Note that $\alpha(t_0) = s_0 = 0$.  

The algorithm will compute quantities
\[V(i, j) = \int_{0}^{t_j} C(f(\alpha(\tau)), g(\tau)) \sqrt{\alpha^\prime(\tau)} \,d \tau,\]
where $\alpha$ is an optimal parametrization with a constraint $\alpha(t_j) = s_i$, which maximizes this integral.

For $k, p \geq 1$, define the following
\[F_k(i, j) = \sqrt{(t_j - t_{j-1}) \sum_{\ell = 0}^{k-1} (s_{i-\ell} - s_{i - \ell -1}) C(f_{i-\ell-1}, g_{j-1})^2},\]
\[G_p(i, j) = \sqrt{(s_i - s_{i-1}) \sum_{\ell = 0}^{p-1} (t_{j-\ell} - t_{j - \ell -1}) C(f_{i-1}, g_{j-\ell-1})^2}.\]
Note that $F_1(i,j) = G_1(i,j)$.
 
By Propositions \ref{caseB}, \ref{caseA}, we have the following recurrence relations:
\[V(i,j) = \max\limits_{k, p} \left\{ V(i-k,j-1) + F_k(i, j), \ V(i-1, j-p) + G_p(i, j) \right\},\]
where the maximum is taken over $1\leq k < i$ if $j > 1$ and $k = i$ if $j=1$; $2\leq p < j$ if $i > 1$, $p = j$ if $i=1$.
The basis of recurrence is $V(0,0) = 0$, and the maximum value of the integral (\ref{integral}) is given by $V(n, m)$.

Note that this algorithm is also applicable to the square root velocity framework.

Let us analyze the complexity of the algorithm. For a fixed $(i, j)$, the complexity of computing $F_k(i, j)$ and $G_p(i, j)$ for all $k$ and $p$ is $n + m$, taking advantage of a simple relation between $F_{k+1}(i, j)$ and $F_k(i,j)$, and similarly for $G$.

Thus, the complexity of the Elastic Time Warping algorithm is $n m (n+m)$.  The memory requirement of this algorithm is $n m$.

\end{document}